\documentclass[runningheads,envcountsame,11pt]{llncs}
\usepackage[utf8]{inputenc}
\usepackage{amsmath}
\usepackage{hyperref}

\newcommand{\OC}{\mathrm{OC}}

\begin{document}
\title{Characteristic Parameters and Special Trapezoidal Words}
%
%
\author{Alma D'Aniello\inst{1}
\and
Alessandro De Luca\inst{2}
}
\authorrunning{A. D'Aniello \and A. De Luca}
%
\institute{Dipartimento di Matematica e Applicazioni
``R.~Caccioppoli'',\\Università degli Studi di Napoli Federico II, Italy\\
\email{alma.daniello@unina.it}
\and
DIETI, Università degli Studi di Napoli Federico II\\
via Claudio 21, 80125 Napoli, Italy\\
\email{alessandro.deluca@unina.it}}
\maketitle              
\begin{abstract}
Following earlier work by Aldo de Luca and others, we study trapezoidal
words and their prefixes, with respect to their characteristic parameters
$K$ and $R$ (length of shortest unrepeated suffix, and shortest length
without right special factors, respectively), as well as their symmetric
versions $H$ and $L$. We consider the distinction between closed (i.e.,
periodic-like) and open prefixes, and between Sturmian and non-Sturmian
ones. Our main results characterize right special and strictly
bispecial trapezoidal words, as done by de Luca and Mignosi for Sturmian words.

\end{abstract}
\section{Introduction}
Sturmian words are certainly among the most studied objects in combinatorics on
words, thanks to their natural definition, interesting characterizations, and
numerous applications in several fields; see~\cite{ch2acow,rigoFLANS} for surveys.
An infinite word is Sturmian if it has exactly $n+1$ distinct \emph{factors}
(blocks of consecutive letters) for each length $n\geq 0$.

Trapezoidal words, introduced in~\cite{daless,delufin}, are a natural finite
analogue of Sturmian words. They have at most $n+1$ factors of each length $n$,
so that the graph of their \emph{factor complexity} function is in the shape
of an isosceles trapezoid (or triangle), whence their name.

The original definition of trapezoidal words, however, uses
\emph{characteristic parameters} $K$ and $R$. For a finite word $w$, 
$K_w$ denotes the length of the shortest unrepeated suffix of $w$,
whereas $R_w$ denotes the smallest integer $n\geq 0$ such that $w$ has no
\emph{right special} factor of length $n$.
A word $w$ is then trapezoidal if and only if its length $|w|$ verifies
$|w|=K_w+R_w$ (with $|w|\geq K_w+R_w$ true in general,
see~\cite{delufin}). Finite Sturmian words, i.e., factors of Sturmian words,
are trapezoidal, but there exist non-Sturmian trapezoidal words such as $aabb$.

In~\cite{trapezo}, the property of being closed (aka \emph{periodic-like}) or
open is considered for trapezoidal words. In particular, the case of (prefixes
of) the Fibonacci word was completely characterized, and this was later
(cf.~\cite{OCtcs17}) extended to all characteristic Sturmian words. Our first
aim, explored in Section~\ref{sec:OC}, is to extend some of those arguments to
the general trapezoidal case, with respect to the values of characteristic
parameters for closed and open prefixes.

In~\cite{delustr}, a characterization was given for right (resp.~left) special
Sturmian words, i.e., finite words $w$ over $A=\{a,b\}$ such that both extensions
$wa,wb$ (resp.~$aw,bw$) are Sturmian. Previously, in~\cite{delumig},
\emph{strictly bispecial} ones (that is, words $w$ such that
$awa,awb,bwa,bwb$ are all Sturmian) had been characterized; in particular,
these turn out to be the noteworthy family of \emph{central} words.
In~\cite{ficiBisp}, bispecial (i.e., simultaneously right and left special)
Sturmian words were characterized.
Our main objective in Section~\ref{sec:sp} is to give similar characterizations
for special trapezoidal words. Special words in a language are often useful for
dealing with enumerative and structural questions.

\section{Notation and Preliminaries}
Let $A=\{a,b\}$ be an alphabet. The free monoid of all words over $A$ under
concatenation is denoted by $A^*$; its neutral element is the \emph{empty word}
$\varepsilon$. For $w\in A^*$ and $x\in A$, $|w|_x$ denotes the number of
occurrences of $x$ in $w$.

If $w=pvs\in A^*$, we say that $v$ is a \emph{factor} of $w$,
$p$ is a \emph{prefix}, and $s$ is a suffix.
A \emph{border} of $w$ is a word that is simultaneously a proper prefix and
a suffix of $w$. The definitions of factor and
prefix also apply to (right-)\emph{infinite words} over $A$.
A word $u$ is a \emph{right} (resp.~\emph{left}) \emph{special} factor of
a finite or infinite word $w$ over $A$ if $ua$ and $ub$ (resp.~$au$, $bu$)
are both factors of $w$.

As anticipated above, an infinite word is \emph{Sturmian} if it has exactly
$n+1$ distinct factors
of each length $n\geq 0$. Therefore Sturmian words are the simplest aperiodic
words in terms of \emph{factor complexity}, and this is one of the reasons of
interest in their study
(see~\cite{ch2acow,rigoFLANS}). Equivalently, a binary infinite
word is Sturmian if and only if it has exactly one right (resp.~left)
special factor of each length. In particular, a Sturmian word is called
\emph{standard} or \emph{characteristic} if all its left special factors
occur as prefixes.

Among the many known characterizations of factors of Sturmian words, or
\emph{finite Stumian words}, perhaps the most famous and widely used one
deals with \emph{balance}; $w\in A^*$ is Sturmian if and only if there is
no word $u$ such that $aua$ and $bub$ are both factors of $w$. Such a pair
of factors for a non-Sturmian word is called a \emph{pathological pair}
(cf.~\cite{trapezo,daless}).

\emph{Central words} are palindromic prefixes of characteristic Sturmian
words. They enjoy many equivalent definitions and interesting properties
(cf.~\cite{delustr,ch2acow}). In particular, a word is central if and
only if it can be written as $a^n$, $b^n$, or $uabv=vbau$ for some
integer $n\geq 0$ and $u,v\in A^*$; in the latter case, $u$ and $v$ are
central words themselves.

The parameters $K_w$ and $R_w$ defined in the previous section were
introduced in~\cite{delufin}, along with their ``left'' counterparts $H_w$
(length of the shortest unrepeated prefix) and $L_w$ (smallest $n\geq 0$ such
that $w$ has no left special factor of length $n$). As already stated, a
finite word is \emph{trapezoidal} if $|w|=K_w+R_w$, or equivalently if
$|w|=H_w+L_w$ (cf.~\cite{delufin,daless}).

Another noteworthy parameter is the (minimal) \emph{period} $\pi_w$ of a word
$w$, which can be defined by $\pi_w=|w|-|v|$ where $v$ is the longest border
of $w$. Central words can also be characterized in terms of periods;
$w\in A^*$ is central if and only if $|w|=\pi_{wa}+\pi_{wb}-2$
(cf.~\cite{delustr}). It is known (see for example~\cite{chst}) that
\emph{periodic extensions} of a finite Sturmian word $w$, i.e., words $w'$
such that $w$ is a factor of $w'$ and $\pi_{w'}=\pi_w$, are still Sturmian;
however, this property does not extend to trapezoidal words.

\begin{example}
\label{ex:pex}
Let $w=aaababa$. Then $H_w=R_w=3$ and $L_w=K_w=4$, so that $w$ is trapezoidal.
Its period is $\pi_w=6$, but the periodic extension $w'=abaaababa=abw$ is not
trapezoidal, as $|w'|=9 > 4+4=K_{w'}+R_{w'}$.
\end{example}

The following theorem is essentially a restatement of~\cite[Theorem~5]{daless}.
It characterizes non-Sturmian trapezoidal words as products of two periodic
extensions of the elements in a pathological pair.
\begin{theorem}
\label{thm:daless}
A word $w\in A^*$ is trapezoidal non-Sturmian if and only if it can be written as
\[w=pxux\cdot yuyq\]
where $u$ is a central word, $A=\{x,y\}$, $\pi_{pxux}=\pi_{ux}$, and
$\pi_{yuyq}=\pi_{yu}$.
Furthermore, $R_w=|pxux|$ and $K_w=|yuyq|$.
\end{theorem}
Note that for such a word, $\{xux,yuy\}$ is actually the \emph{shortest}
pathological pair (cf.~\cite{trapezo}). 

\begin{example}
The trapezoidal word $w=aaababa$ considered in Example~\ref{ex:pex} can be
written as $aaa\cdot baba$, where $aaa$ and $baba$ are periodic extensions
(to the left and to the right, respectively) of the elements of the pathological
pair $\{aaa,bab\}$.

The non-trapezoidal word $w'=abaaababa$ does not verify the condition in
Theorem~\ref{thm:daless}. Indeed, its only pathological pair is $(aaa,bab)$,
and writing $w'=abaaa\cdot baba$ we obtain $\pi_{abaaa}=4\neq 1=\pi_{aa}$.
\end{example}

\section{Closed and Open Trapezoidal Words}
\label{sec:OC}
A finite, nonempty word $w$ is said to be \emph{closed} (or periodic-like in
earlier works) if it has a border $u$ with no internal occurrences, that is, a
factor occurring exclusively as a prefix and as a suffix; another common
terminology for describing this situation is that $w$ is a \emph{complete
(first) return} to $u$. In particular, single letters are closed, their border
being the empty word.

A non-closed word is said to be \emph{open}. Equivalently, $w$ is open if
and only if its longest repeated prefix (resp.~suffix) is a right (resp.~left)
special factor (cf.~\cite{delufin}).

The following result was proved in~\cite[Proposition~3.6]{richper}.
\begin{proposition}
\label{thm:ctrapSt}
All closed trapezoidal words are Sturmian.
\end{proposition}

The following result, showing a basic connection between the behavior of $H$
and the property of being closed or open, is essentially known
(see~\cite[Lemma~6 and Remark~8]{OCtcs17}). We report a proof for the sake of
completeness.
\begin{lemma}
\label{thm:H}
Let $w\in A^{*}$ and $x\in A$. Then $H_{wx}=H_{w}+1$ if $wx$ is closed, and
$H_{wx}=H_{w}$ if $wx$ is open.
\end{lemma}

\begin{proof}
Trivial if $w=\varepsilon$. Let then $w$ be nonempty, and $hy$ ($y\in A$)
be its shortest unrepeated prefix, so that $H_w=|hy|$.
If $wx$ is closed, then its longest border has to be longer than $h$ since
$h$ has internal occurrences in $wx$, and not longer than $hy$ since
otherwise $hy$ would reoccur in $w$. Hence $x=y$ and $H_{wx}=|hy|+1=H_w+1$
as desired.
If $wx$ is open, then $hy$ cannot have internal occurrences in $wx$, since it
is unrepeated in $w$, and it cannot be a suffix either, otherwise $wx$ would
be closed. Hence $hy$ is unrepeated in $wx$, i.e., $H_{wx}=H_w$.
\qed
\end{proof}

\begin{lemma}
\label{thm:L}
Let $wx$ be a trapezoidal word, $x\in A$. Then $L_{wx}=L_{w}$ if $wx$ is closed,
and $L_{wx}=L_{w}+1$ if $wx$ is open.
\end{lemma}
\begin{proof}
Follows from Lemma~\ref{thm:H} as $|w|=H_{w}+L_{w}$ and $|wx|=H_{wx}+L_{wx}$.
\qed
\end{proof}

Clearly, the following symmetric statement holds for \emph{left} extensions $xw$.
\begin{lemma}
\label{thm:K}
Let $w\in A^{*}$ and $x\in A$. Then $K_{xw}=K_{w}+1$ if $xw$ is closed, and
$K_{xw}=K_{w}$ if $xw$ is open. Moreover, if $xw$ is trapezoidal,
then $R_{xw}=R_w$ if $xw$ is closed, and $R_{xw}=R_w+1$ otherwise.
\end{lemma}

For a trapezoidal word $w$, the equality $\{H_w,L_w\}=\{K_w,R_w\}$ holds
(cf.~\cite{delufin}).
The following theorem, proved in~\cite[Proposition~4.4]{trapezo}, is more
precise.
\begin{theorem}
\label{thm:openH=R}
Let $w$ be a trapezoidal word. Then $H_{w}=K_{w}$ and $L_{w}=R_{w}$ if $w$ is
closed, whereas $H_{w}=R_{w}$ and $L_{w}=K_{w}$ if $w$ is open.
\end{theorem}
\begin{corollary}
\label{thm:RKtrap}
Let $wx$ be a trapezoidal word and $x\in A$. Then $K_{wx}=K_{w}+1$ and
$R_{wx}=R_{w}$, unless
\begin{itemize}
\item $w$ is closed and $wx$ is open, or
\item $w$ is open and $wx$ is closed,
\end{itemize}
in which cases we have $K_{wx}=R_{w}+1$ and $R_{wx}=K_{w}$ instead.
\end{corollary}
\begin{proof}
Consequence of Lemmas~\ref{thm:H},~\ref{thm:L}, and
Theorem~\ref{thm:openH=R}.
\end{proof}

\begin{proposition}
Let $w\in A^*$, $y\in A$. If $wy$ is trapezoidal but not Sturmian, then $w$
is open.
\end{proposition}
\begin{proof}
If $w$ is not Sturmian, by Proposition~\ref{thm:ctrapSt} we are done.
Let then $w$ be Sturmian and assume it is closed, by contradiction.
By Proposition~\ref{thm:ctrapSt}, $wy$ is open. Writing
$wy=pxux\cdot yuyq$ as in Theorem~\ref{thm:daless}, we have $q=\varepsilon$
as $w$ is Sturmian, and $H_w=H_{wy}=R_{wy}=|pxux|$ by Lemma~\ref{thm:H} and
Theorems~\ref{thm:openH=R} and~\ref{thm:daless}. It follows that $pxu$ is 
the longest border of $w$. As $x\neq y$ and $w$ ends in $yu$, this is clearly
absurd. \qed
\end{proof}

Let $w_{[n]}$ denote the prefix of $w$ of length $n$.
The \emph{oc-sequence} of a word $w$ is the characteristic sequence of
its closed prefixes. In other terms, it is the binary word $OC_w$ such that
\[OC_w(n)=
\begin{cases}
1\quad\text{ if $w_{[n]}$ is closed,}\\
0\quad\text{ if $w_{[n]}$ is open.}
\end{cases}
\]
The oc-sequence is a useful tool in studying the structure of finite and
infinite words. For example, in~\cite{OCtcs17}, the following was proved:
\begin{theorem}
\label{thm:1k0k}
Let $w$ be an infinite word, and let
\[\OC_w=\prod_{n=0}^{\infty}1^{k_n}0^{k_n'}\]
for suitable positive integers $k_n,k_n'$, with $n\geq 0$.
Then $k_n\leq k_n'$ for all $n\geq 0$, with equality holding for all $n$ if
and only if $w$ is a characteristic Sturmian word.
\end{theorem}

In terms of oc-sequences, an immediate consequence of Lemmas~\ref{thm:H}
and~\ref{thm:L} is the following (see also~\cite[Remark~8]{OCtcs17}):
\begin{proposition}
\label{thm:HLOC}
For any word $w$, $H_w$ is the number of closed non\-empty prefixes of $w$,
i.e., $H_w=|\OC_w|_1$. If $w$ is trapezoidal, then $L_w=|\OC_w|_0$.
\end{proposition}

%

The following two results show the behavior of characteristic parameters $H$ and
$L$ at the end of runs of 1 and 0 in the oc-sequence.
\begin{proposition}
\label{thm:lastC}
Let $w\in A^*$ and $x\in A$ be such that $wx$ is an open trapezoidal word, while
$w$ is closed. Then $L_w<H_w$.
\end{proposition}
\begin{proof}
Since $L_{wx}=L_w+1$ by Lemma~\ref{thm:L}, the longest left special factor
$\ell x$ of $wx$ occurs as a suffix, and $\ell$ is the longest left special
factor of $w$. Clearly, the suffix $\ell$ has internal occurrences in $w$, so
that it is strictly shorter than the longest border $v$. This proves 
$L_w-1=|\ell|<|v|=H_w-1$, whence the assertion.
\qed
\end{proof}

\begin{proposition}
\label{thm:lastO}
Let $w\in A^*$ and $x\in A$ be such that $wx$ is a closed trapezoidal word, while
$w$ is open. Then $H_w\leq L_w$.
\end{proposition}
\begin{proof}
Since $H_{wx}=H_w+1$ by Lemma~\ref{thm:H}, the longest repeated prefix $v$ of $w$
occurs as a suffix, so that $H_w\leq K_w$. The assertion follows by
Theorem~\ref{thm:openH=R}.\qed
\end{proof}

While Theorem~\ref{thm:1k0k} gives local
constraints for an oc-sequence (namely, each run of 1s is followed by a longer
or equal run of 0s), our last three results
can be viewed as more global constraints in the case of trapezoidal words.
Considering the integer parameter
\[D_w:=H_{w}-L_{w}=|\OC_w|_1-|\OC_w|_0\]
gives an interesting way to picture this situation. Indeed
by Proposition~\ref{thm:HLOC}, if $w$ is trapezoidal then
$D_w=|\OC_w|_1-|\OC_w|_0$, so that $D$ increases or
decreases by 1 at each subsequent prefix, depending on whether it is closed or
open; moreover by Propositions~\ref{thm:lastC}--\ref{thm:lastO}, $D$
is necessarily positive (resp.~non-positive) when encountering the last
closed (resp.~open) prefix in a run.

\begin{example}
Let $w=baabaababab$. Then $w_{[n]}$ is closed for $n=1$ and $4\leq n\leq 8$,
and open otherwise; that is, $\OC_w=10011111000$. As predicted by
Propositions~\ref{thm:lastC}--\ref{thm:lastO}, $D$ reaches its (positive) local
maxima, respectively 1 and 4, at $n=1$ and $n=8$, and its (non-positive) local
minimum of $-1$ for $n=3$.
Since $w$ is not Sturmian, by Proposition~\ref{thm:ctrapSt} any subsequent
trapezoidal right extension will be open, leading to an indefinite decrease of
$D$.
\end{example}

\section{Special Trapezoidal Words}
\label{sec:sp}
In analogy with the case of finite Sturmian words (cf.~\cite{delumig,delustr}),
we say that a trapezoidal word $w\in A^*$ is \emph{right} (resp.~\emph{left})
\emph{special} if $wa$ and $wb$ (resp.~$aw,bw$) are both trapezoidal, and that
$w$ is strictly bispecial if $awa,awb,bwa$, and $bwb$ are all trapezoidal.
\begin{proposition}
\label{thm:rspSt}
A right special trapezoidal word is Sturmian.
\end{proposition}
\begin{proof}
Let $w$ be a non-Sturmian trapezoidal word, then open by
Proposition~\ref{thm:ctrapSt}.
If $z\in A$ and $wz$ is trapezoidal, then it is also
not Sturmian (like $w$) and hence open. By Corollary~\ref{thm:RKtrap},
$R_{wz}=R_w$ and $K_{wz}=K_w+1$. By Theorem~\ref{thm:daless}, it follows
$wz=pxux\cdot yuyqz$ with $\pi_{yuyqz}=\pi_{yu}=\pi_{yuyq}$. This shows that
$z$ is uniquely determined, so that $w$ cannot be right special.
\qed
\end{proof}
Symmetrically, one can prove that
\begin{proposition}
\label{thm:lspSt}
A left special trapezoidal word is Sturmian.
\end{proposition}

\begin{theorem}
\label{thm:rspTr}
A trapezoidal word $w$ is right special if and only if either of the
following holds:
    \begin{enumerate}
        \item $w$ is a suffix of a central word, or
        \item $w=pxuxyu$ for a central word $u$, distinct letters $x,y$,
        and a word $p$ such that
        $\pi_{pxux}=\pi_{ux}$.
    \end{enumerate}
Symmetrically, $w$ is a left special trapezoidal word if and only if it is
either a prefix of a central word, or written as $w=uxyuyq$ for $x,y\in A$,
$x\neq y$, and $\pi_{yuyq}=\pi_{yu}$.
\end{theorem}
\begin{proof}
As is well known (cf.~\cite{delustr}), both extensions $wa,wb$ of a word $w$
are Sturmian if and only if $w$ is a suffix of a central word.
Let now $w$ be right special and such that one extension is not Sturmian.
By Proposition~\ref{thm:rspSt}, $w$ is Sturmian. As a consequence of
Theorem~\ref{thm:daless}, we must have $w=pxux\cdot yu$ where $A=\{x,y\}$,
$u$ is some central word, and $p$ is such that $\pi_{pxux}=\pi_{ux}$.

Conversely, if $w=pxuxyu$ with $A=\{x,y\}$, $u$ central and
$\pi_{pxux}=\pi_{ux}$, then $w$ is Sturmian as $\{xux,yuy\}$ is the only
pathological pair in the trapezoidal non-Sturmian word $wy$; therefore, $wx$
must be Sturmian (and then trapezoidal) too.

The left special case is similar.
\qed
\end{proof}

The following theorem is a restatement of results in~\cite{ficiBisp,delustr}; it
characterizes Sturmian words that are bispecial (\emph{as Sturmian words}).
\begin{theorem}
\label{thm:fbisp}
Let $w\in A^*$. Then $wa,wb,aw,bw$ are all Sturmian if and
only if $w=(uxy)^nu$ for some central word $u$, $\{x,y\}=A$ and a
nonnegative integer $n$.
Furthermore, $awa,awb,bwa,bwb$ are all Sturmian if and only if $w$ is central,
i.e., $n=0$, whereas for $n>0$ exactly one such bilateral extension is not
Sturmian, namely $xwy$.
\end{theorem}

\emph{Semicentral} words were defined in~\cite{trapezo} by the property of
having their longest repeated prefix, longest repeated suffix, longest left
special factor, and longest right special factor coincide.
In the same paper, they were characterized as words $w$ such that $w=uxyu$ for
some central word $u$ over $A=\{x,y\}$. Hence, they correspond to the case
$n=1$ in the previous theorem.

Our final result is a characterization of strictly bispecial trapezoidal
words.
\begin{theorem}
A trapezoidal word is strictly bispecial if and only if it is central or
semicentral.
\end{theorem}
\begin{proof}
By Theorem~\ref{thm:fbisp}, central words are strictly bispecial.
Moreover, by the same theorem all bilateral extensions of a semicentral
word $uxyu$ are Sturmian, except for $xuxyuy$ which is trapezoidal non-Sturmian
by Theorem~\ref{thm:daless}.

Conversely, if $w$ is a strictly bispecial trapezoidal word, then either all
bilateral extensions are Sturmian, in which case $w$ is central by
Theorem~\ref{thm:fbisp} and we are done, or at least one is not.

Assume, for instance, that
$cwa$ is trapezoidal non-Sturmian, the other cases being similar. By
Proposition~\ref{thm:rspSt}, $cw$ is Sturmian, so that $cwb$ must be too.
Symmetrically, as a consequence of Proposition~\ref{thm:lspSt}, $dwa$ must be
Sturmian as well (where $\{c,d\}=A$).
In all cases, $wa,wb,aw,bw$ are all Sturmian. By Theorem~\ref{thm:fbisp}, it
follows $w=(uxy)^nu$ for some $n>0$; as a consequence of
Theorem~\ref{thm:daless}, we must have $n=1$.
\qed
\end{proof}

\section{Concluding Remarks}
A few related problems remain open. In particular, in~\cite{OCtcs17} the
oc-sequence for (prefixes of) characteristic Sturmian words was characterized,
see Theorem~\ref{thm:1k0k}. The general
trapezoidal case, and even the non-standard Sturmian one, is still open.
We believe our results may shed some light on the matter, as illustrated at the
end of Section~\ref{sec:OC}.

Regarding the preceding section, a simple and elegant characterization of
(not necessarily strictly) bispecial trapezoidal words, such as
Theorem~\ref{thm:fbisp} is for the Sturmian case, is still missing.
Theorem~\ref{thm:rspTr} might be an ingredient for such a result.

\section*{Acknowledgments}
We thank the anonymous referees for their many helpful comments.
This paper is dedicated to the memory of our dear colleague
Aldo de Luca (1941--2018).


\end{document}